\documentclass[a4paper,UKenglish,cleveref, autoref]{article}
\usepackage{times}  % DO NOT CHANGE THIS
\usepackage{helvet} % DO NOT CHANGE THIS
\usepackage{courier}  % DO NOT CHANGE THIS
\usepackage{url}  % DO NOT CHANGE THIS
\usepackage{graphicx}  % DO NOT CHANGE THIS
\usepackage{hyperref}
\usepackage{comment}
\usepackage[utf8]{inputenc}
\usepackage[english]{babel}
\usepackage{amsmath}
\usepackage{amssymb}
\usepackage{amsthm}
\usepackage{thm-restate}
\usepackage{vmargin}
\setmarginsrb{1in}{1in}{1in}{1in}{0pt}{0pt}{0pt}{6mm}

\usepackage{algorithm}
\usepackage{algorithmic}
\usepackage{xspace}
\usepackage{xargs}

\usepackage[switch]{lineno}

\usepackage{amssymb,amsmath,amsthm}
\usepackage{graphicx,color}
\usepackage{nicefrac}
\usepackage{todonotes}

\newtheorem{theorem}{Theorem}
\newtheorem{corollary}{Corollary}
\newtheorem{definition}{Definition}

\newtheorem{lemma}{Lemma}
\newtheorem{claim-n}{Claim}[section]
\newtheorem{remark}{Remark}

\DeclareMathOperator*{\argmin}{arg\,min}

\newcommand{\costirrationality}{X_\beta}

\newcommand{\Oh}{\mathcal{O}}

\newcommand{\Em}{\mathbf{E}}

\newcommand{\Var}{\mathbf{Var}}
\newcommand{\EstIrrCost}{\textsc{Estimating the Cost of Irrationality}\xspace}
\newcommand{\MinIrrCost}{\textsc{Minimum Cost of Irrationality}\xspace}
\newcommand{\MaxIrrCost}{\textsc{Maximum Cost of Irrationality}\xspace}
\newcommand{\CountPar}{\textsc{Counting Partitions}\xspace}
\newcommand{\EstIrrCostshort}{\textsc{ECI}\xspace}
\newcommand{\MCI}{\textsc{MCI}\xspace}

\newcommand{\tw}{\operatorname{tw}}
\newcommand{\fvs}{\operatorname{fvs}}
\newcommand{\vc}{\operatorname{vc}}
\newcommand{\fes}{\operatorname{fes}}
\DeclareMathOperator{\operatorClassNP}{{\sf NP}}
\newcommand{\classNP}{\ensuremath{\operatorClassNP}}

\DeclareMathOperator{\operatorClassFPT}{{\sf FPT}\xspace}
\newcommand{\classFPT}{\ensuremath{\operatorClassFPT}\xspace}
\DeclareMathOperator{\operatorClassW}{{\sf W}}
\newcommand{\classW}[1]{\ensuremath{\operatorClassW[#1]}}

\DeclareMathOperator{\operatorClassXP}{{\sf X}P\xspace}
\newcommand{\classXP}{\ensuremath{\operatorClassXP}\xspace}
\DeclareMathOperator{\operatorClassSharpP}{ {\#} {\sf P}\xspace}
\newcommand{\classSharpP}{\ensuremath{\operatorClassSharpP}\xspace}

\newcommandx{\defparproblem}[4][]{
    \vspace{3mm}
    \noindent\fbox{
        \begin{minipage}{0.8\textwidth}
            \begin{tabular*}{\textwidth}{@{\extracolsep{\fill}}lr} {#1} \\ 
            \end{tabular*}
            {\textbf{Input:}} {#2}  \\
            {\textbf{Parameter:}} {#3}  \\
            {\textbf{Task:}} {#4}
        \end{minipage}
    }
    \vspace{4mm}
}
\newcommandx{\defsimpleproblem}[3][]{
    \vspace{3mm}
    \noindent\fbox{
        \begin{minipage}{0.8\textwidth}
            \begin{tabular*}{\textwidth}{@{\extracolsep{\fill}}lr} {#1} \\ 
            \end{tabular*}
            {\textbf{Input:}} {#2}  \\
            {\textbf{Task:}} {#3}
        \end{minipage}
    }
    \vspace{4mm}
}
\title{Inconsistent Planning: When in doubt, toss a coin!}
\author{Yuriy Dementiev,\textsuperscript{\rm 1} Fedor V.~Fomin\footnote{The research leading to these results has received funding from the Research Council of Norway via the project  BWCA (grant no. 314528).},\textsuperscript{\rm 2} Artur Ignatiev\textsuperscript{\rm 1}\\
\textsuperscript{\rm 1} St.Petersburg State University, Russia\\
\textsuperscript{\rm 2} Department of Informatics, University of Bergen, Norway\\
yuru.dementiev@gmail.com, fedor.fomin@uib.no, artur.ignatev23924@gmail.com}
\date{}
\begin{document}

\maketitle
%%%%%%%%%%%%%%%%%%%%% INTRODUCTION
\begin{abstract}
    One of the most widespread human behavioral biases is the present bias---the tendency to overestimate current costs by a bias factor. Kleinberg and Oren (2014) introduced an elegant graph-theoretical model of inconsistent planning capturing the behavior of a present-biased agent accomplishing a set of actions. The essential measure of the system introduced by Kleinberg and Oren  is the \emph{cost of irrationality}---the ratio of the total cost of the actions performed by the present-biased agent to the optimal cost. This measure is vital for a task designer to estimate the aftermaths of human behavior related to time-inconsistent planning, including procrastination and abandonment.  

As we prove in this paper, the cost of irrationality is highly susceptible to the agent's choices when faced with a few possible actions of equal estimated costs. To address this issue, we propose a modification of Kleinberg-Oren's model of inconsistent planning.  In our model, when an agent selects from several options of minimum prescribed cost,  he uses a randomized procedure. We explore the algorithmic complexity of computing and estimating the cost of irrationality in the new model. 
\end{abstract}
\section{Introduction}

\emph{Time-inconsistent behavior} is the term in behavioral economics and psychology describing the behavior of an agent optimizing a course of future actions but changing his optimal plans in the short run without new circumstances \cite{thaler2015misbehaving}. For example, why do we buy a year swim membership and not go to the swimming pool after that? Why do we procrastinate when it comes to paying off credit card debt? Why do we want to eat healthier but have little incentive to do so?  As  Socrates in Plato's Protagoras asks, if one judges a certain behavior to be the best course of action, why would one do anything else?

A standard assumption in behavioral economics used to explain the gap between long-term intention and short-term decision-making is the notion of \emph{present bias}. According to  \cite{o1999doing}, when considering trade-offs between two future moments, present-biased preferences give stronger relative weight to the earlier moment as it gets closer.  

The mathematical idea of present bias goes back to 1937 when \cite{Samuelson1937} introduced the discounted-utility model. It has developed into the hyperbolic discounting model, one of the cornerstones of behavioral economics~\cite{Laibson1994,McClure2004}. 
A simple mathematical model of  
    present bias was suggested in~\cite{Akerlof91}.
    In Akerlof's model,  the  {\  salience factor} causes the agent to put more weight on immediate events than on the future. Thus 
the cost of an action that will be  perceived in the future is assumed to be $\beta$ times smaller than its actual cost,  for some present-bias parameter  $\beta <1$. %, called the \emph{degree of present bias}.    
   It appears that even a tiny salience factor could yield high extra costs for the agent.  
 %The model was used for studying various aspects of human behavior related to \emph{time-inconsistent} planning, including \emph{procrastination}, and \emph{abandonment}.
 
Kleinberg and Oren~\cite{KleinbergO14,KleinbergO18} introduced an elegant  graph-theoretic model  encapsulating the salience factor and scenarios of Akerlof.   The approach   is based on analyzing how an agent traverses from a source~$s$ to a target~$t$ in a directed  edge-weighted graph~$G$.   Before defining this model formally,  we provide an  illustrating  example. The example, up to small modifications, is borrowed from ~\cite{KleinbergO14} and originally due to Akerlof~\cite{Akerlof91}.

\medskip\noindent\textbf{Example.} One of the authors of this paper, we call him Bob, is planning to write reviews for AAAI.  He estimates the cost (say, estimated time) of this task as $c$.  While the definite deadline is on Friday, Bob's  initial plan is to write reviews on Monday. However, on Monday, Bob realizes that he also needs to check Google Scholar to find out who cited his paper. He estimates the cost of the other task as $x$.  Now Bob meets the dilemma. He can either (a) write reviews today or (b) check Google Scholar today and write reviews tomorrow. While estimating the costs, Bob uses the present-bias parameter $\beta$. Thus he estimates the cost of (b) as $x+\beta c$. It appears that  $x+\beta c<c$, therefore Bob decides to pursue (b). On Tuesday, the story repeats by procrastinating with Instagram, and on Wednesday with Facebook, see Fig.~\ref{fig:ackerlof}.  At the end, Bob sends the reviews on Friday, spending on this job totally $4x+c$ instead of $c$. 
\begin{figure}[ht]
\center{\includegraphics[scale=0.3]{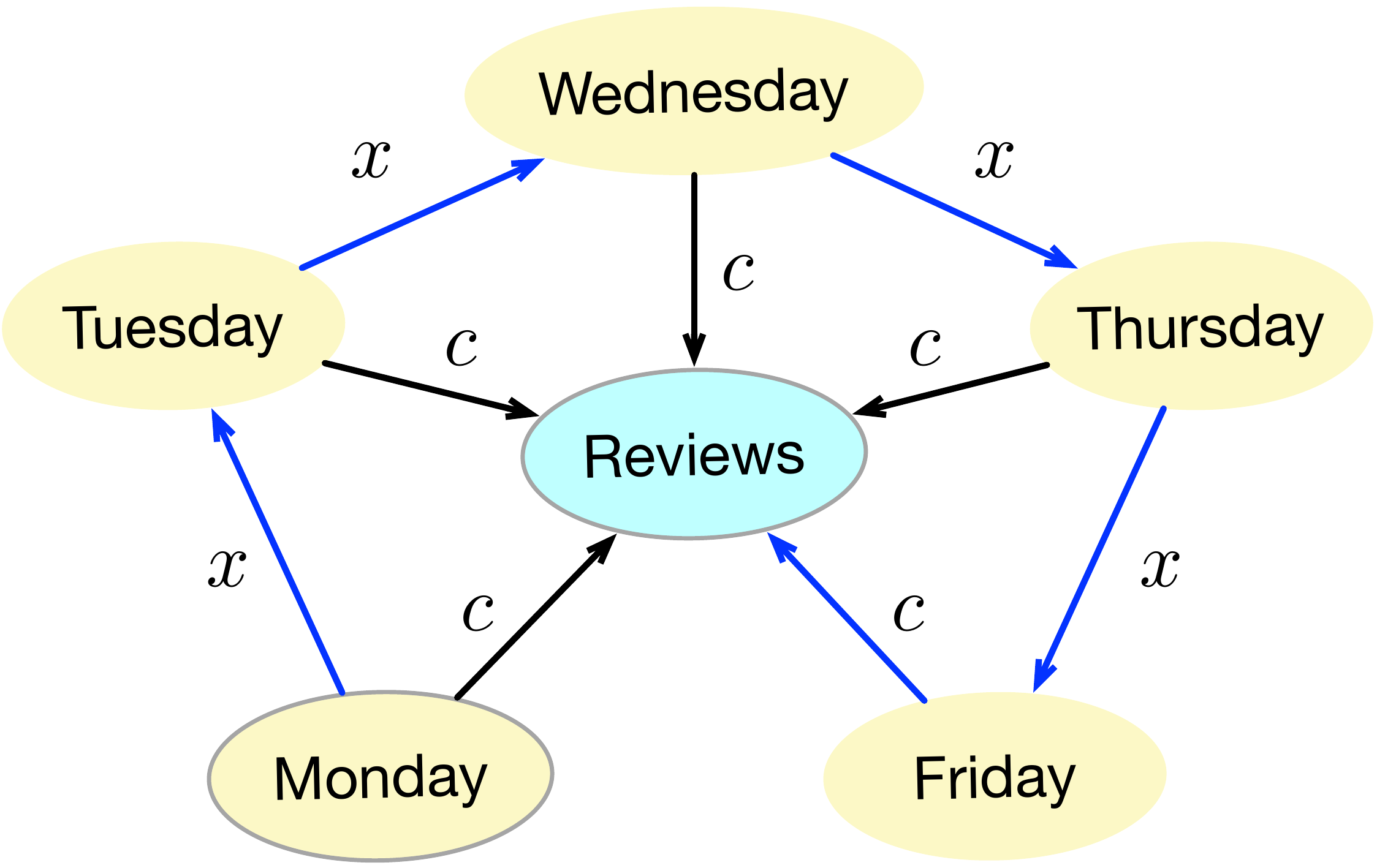}}
\caption{Task graph $G$ with the source point $s=\texttt{Monday}$ and the target point $t=\texttt{Reviews}$. Assuming $x+\beta c<c$, Bob will follow the blue path rather just writing the reviews on Monday. The ratio between the costs of the path followed by the agent and the optimal $s$-$t$ path can be made arbitrarily large by adding more days of procrastination.}\label{fig:ackerlof}
\end{figure}

\medskip\noindent\textbf{Kleinberg-Oren's Model.}
 %We here present the model due to Kleinberg and Oren~\cite{KleinbergO14}. Formally, a
  An instance of the {\em time-inconsistent planning model} is a 5-tuple $M = (G, w, s, t,   \beta)$ where:
\begin{itemize}
    \item $G = (V(G), E(G))$ is a directed acyclic $n$-vertex graph called a {\em task graph}. $V(G)$ is a set of elements called {\em vertices}, and $E(G) \subseteq V(G)\times V(G)$ is a set of \emph{arcs} (directed {edges}). The graph is {\em acyclic}, which means that there exists an ordering of the vertices called a {\em topological order} such that, for each edge, its first endpoint comes strictly before its second endpoint in the ordering.  Informally speaking, vertices represent states of intermediate progress, whereas edges represent possible actions that transitions an agent between states.
    \item $w : E(G) \to \mathbb{N}$ is a function representing the costs of transitions between states. The transition of the agent  from state $u$ to state $v$ along arc $uv\in E(G)$ is of cost  $w(uv)$. 
    %assigning non-negative weight to each arc. %Informally speaking, this is the cost incurred by performing a certain action.
    \item The agent starts from the start vertex $s \in V(G)$. 
    \item $t \in V(G)$ is the target vertex.
    %\item $r \in \mathbb{R}_{\geq 0}$ is the reward.
    \item $\beta  \leq 1$ is the agent's present-bias parameter. 
\end{itemize}

An agent is initially at vertex $s$ and can move in the graph along arcs in their designated direction. The agent's task is to reach the target $t$. The agent moves according to the following rule. When standing at a vertex $v$, the agent evaluates (with a present bias) all possible paths from $v$ to $t$.  In particular, a $v$-$t$ path $P \subseteq G$ with edges $e_1, e_2, \ldots,e_p$ is evaluated by the agent standing at $v$ to cost 
\[\zeta_M(P) = w(e_1) + \beta  \cdot  \sum^{p}_{i=2}w(e_i).\] We refer to this as the {\em perceived} cost of the path $P$. For a vertex $v$, its \emph{perceived cost to the target} is the minimum perceived cost of any path to $t$, \[\zeta_M(v) = \min \{\zeta_M(P) \mid P \text{ is a }v\text{-}t \text{ path}\}.\] We refer to an  $v$-$t$ path $P$ with perceived cost  $\zeta_M(v)$ as to a \emph{perceived path}. 
  Thus when in vertex $v$, the agent picks one of the perceived paths and traverses its first edge, say $vu$.
  After arriving to the new vertex  $u$, the agent  computes the  {perceived cost to the target} $\zeta_M(u)$, selects a perceived   $u$-$t$ path , and   traverse its first edge. This repeats until the agent   reaches  $t$.

 % \todo[inline]{Add example here}
  
%  It is possible, that traversing the graph from $s$ to $t$  the agent can follow 
  Let $P_\beta (s,t)$ be a  $s$-$t$ path followed by an agent with present-bias $\beta$ and let $c_\beta (s,t)$ be the cost of this path.  Let $d(s,t)$ be the distance, that is the cost of a shortest $s$-$t$ path. Then Kleinberg and Oren defined the  measure describing the ``\emph{price of irrationality}'' of the system.
  \begin{definition}[Cost of irrationality \cite{KleinbergO14}]  \label{def:cost1}The \emph{cost of the irrationality} of the  time-inconsistent planning model   $M = (G, w, s, t,   \beta)$ is 
  \[
  \frac{c_\beta (s,t)}{d(s,t)}
  .\]
  \end{definition}
   
Thus in our example in Fig.~\ref{fig:ackerlof}, the cost of irrationality is  $\frac{4x+c}{c}$.  
One omitted detail in the definition makes the meaning of the cost of irrationality ambiguous. 
It could be that several paths with minimum perceived cost $\zeta_M(v)$ lead from $v$. In this situation an agent in the state $v$ might be  indifferent between several arcs leaving  $v$---they both evaluate to equal perceived  costs. Hence 
 there could be several different \emph{feasible} paths $P_\beta (s,t)$ which the agent could follow.  
While for the agent standing in a vertex $v$ the perceived costs of all perceived paths are the same, the actual costs of feasible paths could be different.  

% \todo[inline]{Add example here}

Two approaches to address this ambiguity could be found in the literature.  First, one can assume that an agent uses a consistent tie-breaking rule. For example, 
Kleinberg and Oren in~\cite{KleinbergO14} suggest selecting the node that is earlier in a fixed topological ordering of $G$. Kleinberg, Oren, and Raghavan~\cite{KleinbergOR16} consider the situation when the arcs are ordered, and an agent selects the largest available arc.  The disadvantage of this approach is that it is hard to imagine someone building their plans based on a topological ordering of tasks or arbitrarily labeled arcs in a real-life scenario.  
Another approach taken by Albers and Kraft~\cite{Albers2018}  and by  Fomin and Str{\o}mme~\cite{FominS20} is to break ties arbitrarily. As we will see, depending on how an agent breaks the ties, the cost of irrationality could change exponentially in the number of vertices.  Therefore, with the second approach, the value of the cost of irrationality is not well-defined.

Because of that, we revisit the model of Kleinberg and Oren in \cite{KleinbergO14} and redefine the cost of irrationality. Our approach is natural ---\emph{when in doubt, toss a coin!}  When several paths of minimum prescribed cost lead from $v$, the agent selects one of them with some probability and traverses the first arc of this path.  

More precisely, we view the graph as a Markov decision process. Thus the 
instance of the {\em time-inconsistent planning model} is a 6-tuple $M = (G, w, s, t, p,  \beta)$, where for  each edge $uv$ of the task graph, we assign the probability  $p(u,v)$ of transition $u\to v$. 
Here  for every $u\in V(G)$, 
$\sum_{uv\in E(G)}p(u,v)=1.$

Moreover, the probability can be positive only for edges that could serve for transitions of the agent. In other words, $p(u,v)>0$ only if there is a  $u$-$t$ path $P$ of perceived cost $\zeta_M(u)$ whose first edge is  $uv$. 
The selection of probability $p$ corresponds to some predictions or future preferences in breaking the ties. For example, 
when the agent at stage $u$ faces $\ell$ $u$-$t$ paths of minimum perceived cost and has no preferences over any of them,   it would be natural to assign each transition from $u$ the probability $1/\ell$. On the other hand, if the agent has preferences in selecting from paths of equal costs, this can be controlled by a different selection of   $p$. 
With these settings, we call an $s$-$t$ path $P $ \emph{feasible}, if with a non-zero probability the present-biased agent will follow $P$.

 Now we can define the cost of  the agent with present-bias $\beta$ as  discrete random variable $C_\beta$ with $\Pr (C_\beta =W)$ being the probability that the path traversed  by the agent is of cost $W$.   Then we can redefine the cost of irrationality as follows. 
  
    \begin{definition}[Revised cost of irrationality] \label{def:cost2}The \emph{cost of the irrationality} of the  time-inconsistent planning model   $M = (G, w, s, t, p,  \beta)$ is 
  \[
\costirrationality=  \frac{C_\beta}{d(s,t)}
  .\]
  \end{definition}
 Let us note that when no ties occur, our definition coincides with the definition of Kleinberg and Oren.  
Estimating the cost of irrationality $\costirrationality$ could help the task-designer to evaluate the chances of \emph{abandonment}, the situation when an agent realizes that accomplishing the task takes much more effort than he presumed initially, and thus ultimately gives up.
  
  \medskip\noindent\textbf{Example cont.}
  In the example in Fig.~\ref{fig:ackerlof}, let us put $c=6$, $x=3$, and $\beta = \frac{1}{2}$. We also assume that Bob does not have preferences between two actions of minimum perceived costs and thus pursue one of the actions with probability $p=1/2$. On Monday, Bob selects between two options of perceived costs $6$: either to write reviews that costs $c=6$, or to check Google Scholar and write reviews tomorrow, which costs $x+\beta c=6$. The probability that Bob will finish reviews on Monday, and thus will spend $c=6$ hours,  is $1/2$.   The optimal cost $d(s,t)=c=6$. 
  Hence   $\Pr (C_\beta \leq 6)=\frac{1}{2}$ and thus   $\Pr (\costirrationality \leq 1)=\frac{1}{2}$ . The probability that Bob finishes the job on Tuesday and thus will spend $x+c=9$, is $\big(\frac{1}{2}\big)^2$. Therefore,  $\Pr (\costirrationality  \leq 9/6=3/2)=\frac{1}{2}+\big(\frac{1}{2}\big)^2.$ The situation repeats up till Thursday, and we have that for $1\leq i\leq 4$, 
$\Pr (\costirrationality  \leq 1+ (i-1)/2)=\sum_{j=1}^i  \Big(\frac{1}{2}\Big)^j.$
Bob has to submit by  Friday, so    $\Pr (\costirrationality  \leq 3)=1$.

\medskip\noindent\textbf{Our contribution.}  We introduce the randomized version of the cost of irrationality  and initiate its study from the computational perspective. 
To support our point of view on the cost of irrationality, 
%we show that the number of feasible paths of different costs the agent can pursue is exponential. This also shows, that if we define the cost of irrationality by Definition~\ref{def:cost1}, the difference between different irrationality costs could be exponential in $|V(G)|$. 
%
%
%
 we start from the combinatorial result (Theorem~\ref{thm:expnumpaths}), showing that there are   time-inconsistent planning models
with exponentially (in $n$) many feasible paths of different costs. It yields that in the deterministic model of Kleinberg and Oren (Definition~\ref{def:cost1}) there could be   exponentially many different costs of irrationality. 

To study the cost of irrationality $X_\beta$, we define the following computational problem.  

 \defsimpleproblem{\EstIrrCost (\EstIrrCostshort)}{A time-inconsistent planning model $M = (G, w, s, t, p,  \beta)$, and $W\geq 0$.} {Compute $\Pr (\costirrationality  \leq W)$.}

 We show in Theorem~\ref{theorem_ECIhard} that  \EstIrrCostshort is \classSharpP-hard. %\todo{Change the order of section, this should go before algorithms}
  Thus computationally,   \EstIrrCostshort is not easier than counting Hamiltonian cycles, counting perfect matching, satidying assignments, and all other  \classSharpP-complete problems. Our hardness proof strongly exploits the fact that the edge weight $w$ of the model are exponential in the $n$, the number of vertices of $G$. We show that when the edge weights are bounded by some polynomial of $n$, then  \EstIrrCostshort is solvable in polynomial time. We also obtain polynomial time algorithms, even for exponential weights, for the important ``border'' cases: minimum, maximum,  and average. More precisely, we prove that each of the following tasks
 \begin{itemize}
\item[(a)] finding the  minimum value $W$ such that $\Pr (\costirrationality  \leq W)$ is positive and computing  $\Pr (\costirrationality  \leq W)$ (Theorem~\ref{lemma:MinCI}), 
\item[(b)]   finding the  minimum value $W$ such that  $\Pr (\costirrationality  \leq W)=1$ (Theorem~\ref{lemma:MinCI}), and 
\item[(c)] computing $\Em(X_\beta)$ (Theorem~\ref{thm:expect}), 
\end{itemize}
can be done in polynomial time.
 % \todo[inline]{Add references to items (a), (b), (c). State (b) as Theorem, you can say that the proof is similar to (a). Done.}
 
 We also take a look at  \EstIrrCostshort from the perspective of structural parameterized complexity. Structural parameterized complexity is the common tool in graph algorithms for analyzing intractable problems. Thus we are interested how the structure of the graph $G$ in the time-inconsistent model could be used to design efficient algorithms. For example, the problem of finding a maximum weight set of  independent  vertices is an NP-hard problem. However, it becomes tractable when the treewidth of the input graph is bounded. We will have more discussions about parameterized complexity in  Section~\ref{sectionParCom}. The most popular graph parameters in parameterized complexity are the treewidth of a graph, the size of a minimum feedback vertex set and vertex cover   \cite{CyganFKLMPPS15}. For a directed graph $G$, let $\tw(G)$, $\fvs(G)$, and $\vc(G)$  be the treewidth, the minimum size of a feedback vertex set and the minimum size of a vertex cover of the underlying undirected graph of $G$, correspondingly.  
 The natural interpretation of such parameters is that they in some sense capture the complexity of the decision-making. For example, when the feedback vertex set ($\fvs(G)$) equals to $0$, i.e. if the graph is a directed tree, then the problem becomes very simple---the agent always follows one initial plan. However, the complexity of the problem scales with $\fvs(G)$. Similarly,  the other structural parameters like $\vc(G)$  or $\tw(G)$ could also reflect the complexity of the problem.

 We prove the following 
 \begin{itemize}
\item   \EstIrrCostshort is   \classSharpP-hard even when in the  time-inconsistent planning model $M = (G, w, s, t, p,  \beta)$, we have $\tw(G)=2$. (This result actually follows directly from the reduction of Theorem~\ref{theorem_ECIhard}.)
\item  \EstIrrCostshort is  W[1]-hard parameterized by  $\fvs(G)$ and by $\vc(G)$ (Theorem~\ref{thm:w1-hard}).
On the other hand, \EstIrrCostshort is solvable in times $n^{\Oh (\fvs(G))}$ and    $n^{\Oh (\vc(G))}$.
\end{itemize}
On the other hand,  when parameterized by the minimum size of the feedback edge set of the underlying graph, $\fes(G)$, that is the set of edges whose removal makes the graph acyclic, the problem becomes  fixed-parameter tractable.

% \todo[inline]{Check that everywhere we speak about cost $\costirrationality $ and not $C_\beta$.}

 Our results demonstrate that while computing the  cost of irrationality  is intractable in the worst-case, in many interesting situations this parameter could be computed efficiently.

\medskip\noindent\textbf{Related Work.}

% Note that quasi-hyperbolic discounting (discussed in~\cite{Laibson1994,McClure2004}) can be seen as a generalization of both Samuelson's discounted-continuity model~\cite{Samuelson1937} and Akerlof's salience factor~\cite{akerlof1991procrastination}. There has been some empirical support for this model; however there are also many known psychological phenomena about time-inconsistent behavior it does not capture~\cite{Frederick2002}. 

%  This salience factor also has support from psychology, where McClure et.~al showed by using brain imaging that separate neural systems are in play when humans value immediate and delayed rewards~\cite{McClure2004}.

% Kleinberg and Oren~\cite{KleinbergO14},

% motivated by the earlier work by Akerlof~\cite{Akerlof91}. 
The graph-theoretical model we use in this paper for time-inconsistent planning is due to~\cite{KleinbergO14,KleinbergO18}. We refer to these papers for  a survey of earlier work on time-inconsistent planning, with connections to procrastination, abandonment, and choice reduction.  There is a significant amount of the follow-up work on the  the model of Kleinberg and Oren.   Albers and Kraft~\cite{Albers2018} studied the ability to place rewards at nodes for motivating and guiding the agent. They show hardness and inapproximability results and provide an approximation algorithm whose performances match the inapproximability bound. The same authors considered another approach in~\cite{AlbersK17} for overcoming these hardness issues by allowing not to remove edges but to increase their weight. They were able to design a 2-approximation algorithm in this context. Tang et al.~\cite{tang2017computational} also proved hardness results related to the placement of rewards and showed that finding a motivating subgraph is NP-hard.    Fomin and Str{\o}mme~\cite{FominS20} studied the parameterized complexity of computing a motivating subgraph in the model of Kleinberg and Oren.

Gravin et al.~\cite{GravinILP16arch}   extended the model by considering the case where the degree of present bias may vary over time, drawn independently at each step from a fixed distribution. In particular, they described the structure of the worst-case graph for any distribution. They derived conditions on this distribution under which the worst-case cost ratio is exponential or constant. 
The difference with our work  is that in our model the present bias does not change, while the randomness comes from selecting among paths of equal perceived costs. However, our model can be easily mixed with the model of Gravin et al. by varying the present bias and selection rule at every step.  Another revision of the original model could be found in \cite{KleinbergOR16},  where agents estimate the degree $\beta$ of present bias erroneously---either underestimating or overestimating that degree. The behavior of such agents is  compared with the behavior of  ``sophisticated'' agents who 
are aware of their present-biased behavior in the future. 
 In~\cite{KleinbergOR17},   not only agents suffer from present biases, but also from \emph{sunk-cost} bias, i.e., the tendency to incorporate costs experienced in the past into one's plans for the future.  Albers and Kraft~\cite{AlbersK17b} considered a model with uncertainty, bearing similarities with~\cite{KleinbergOR16}, in which the agent is solely aware that the degree of present bias belongs to some set $B \subset (0, 1]$, and may or may not vary over time. In \cite{FominFG21} agents make their decisions by solving some optimization problem. 
 
 The combinatorial problem related to computing the cost of irrationality is  the problem of counting the number of weighted shortest $s$-$t$ path in an acyclic directed graph.
 The complexity and approximability of this problem was studied by 
Mihal{\'{a}}k et al. 
\cite{MihalakSW16}.

%%%%%%%%%%%%%%%%%%%%% exp_cost_ratio
\section{Exponential number of different cost of irrationality}\label{sec:exp}
In this section we provide an example supporting our definition of the cost of irrationality (Definition~\ref{def:cost2}). In our construction, the agent following from $s$ to $t$ could follow one of exponentially many feasible paths of different final costs.  It implies that  the cost of irrationality in deterministic (Definition~\ref{def:cost1})  could vary exponentially depending on how  the agent selects between paths of equal perceived  costs. The construction we use to prove 
Theorem~\ref{thm:expnumpaths} is also used to obtain the complexity result. 

%  
% we start from the combinatorial result (Theorem~\ref{thm:expnumpaths}), showing that there are   time-inconsistent planning models
%with exponentially (in $n$) many feasible paths of different costs. It yields that in the deterministic model of Kleinberg and Oren (Definition~\ref{def:cost1}) there could be   exponentially many different costs of irrationality. 
%
%
%At the beginning we formulate a simple lemma, which, however, will be used many times in the future.
%\begin{lemma}
%\label{lem1}
%If there is a branch-merge structure (see figure), then the choice of an agent at the vertex of the branch depends only on the parts before the merge.
%\end{lemma}
%
%\begin{figure}[ht]
%\center{\includegraphics[scale=0.13]{picture/lemma1.pdf}}
%\caption{Branch-merge structure}
%\end{figure}
%
%\begin{proof}
%
%Consider the choice of an agent at the fork vertex, to evaluate both the upper and lower paths, it will take the shortest path from the merge vertex to the target vertex, and add its weight with a coefficient of $\beta$, so the choice will depend only on the separated parts of the paths.
%\end{proof}
%
%Now, using this lemma, we can get the following interesting result.

\begin{theorem}\label{thm:expnumpaths}
There is a graph with an exponential (in  the number of vertices) number of feasible paths of different costs.
\end{theorem}

\begin{proof}
First we give an explicit construction of such a graph. 
 %Then we define the  weights and  present bias $\beta$. 
  We construct  a gadget (see Fig.~\ref{Exp-Number}) with  4 vertices. 
From the gadgets we build an $n$-vertex graph by forming a sequential chain of $ \frac{n-1}{3} = \lfloor \frac{n}{3} \rfloor$ of such gadgets, see Fig.~\ref{Exp-Number}. We  arrange the weights on the arcs so that 
for the agent standing in a branching vertex $v$,  
the  {perceived cost to the target} $\zeta_M(v)$ is realized by a path following from $v$ through the two upper arcs of the gadget and through the lower arcs of the gadget. 
 To simplify the notation, we put the weights of the edges in the $i$-th gadget equal to $a_i, b_i, c_i$ and $d_i$. We want the weight satisfy the following equalities $c_i + d_i - (a_i+b_i) = 2^i$ and $a_i + b_i = A$. Then,  for a branching vertex $v$, the  {perceived cost} of every $v$-$t$ paths in the graph will be equal to  $\zeta_M(v)$. However, the real cost of the selected path would be $(c_i + d_i)$ or $(a_i + b_i)$, depending which part of the agent decide to follow. 
 Thus, for any integer $x\in \left[ A \cdot \lfloor \frac{n}{3} \rfloor,\ A \cdot \lfloor \frac{n}{3} \rfloor + 2^{\lfloor \frac{n}{3} \rfloor - 1} \right]$, the agent will be able to choose a path that has a real cost of $x$.

\begin{figure}[ht]
\center{\includegraphics[scale=0.18]{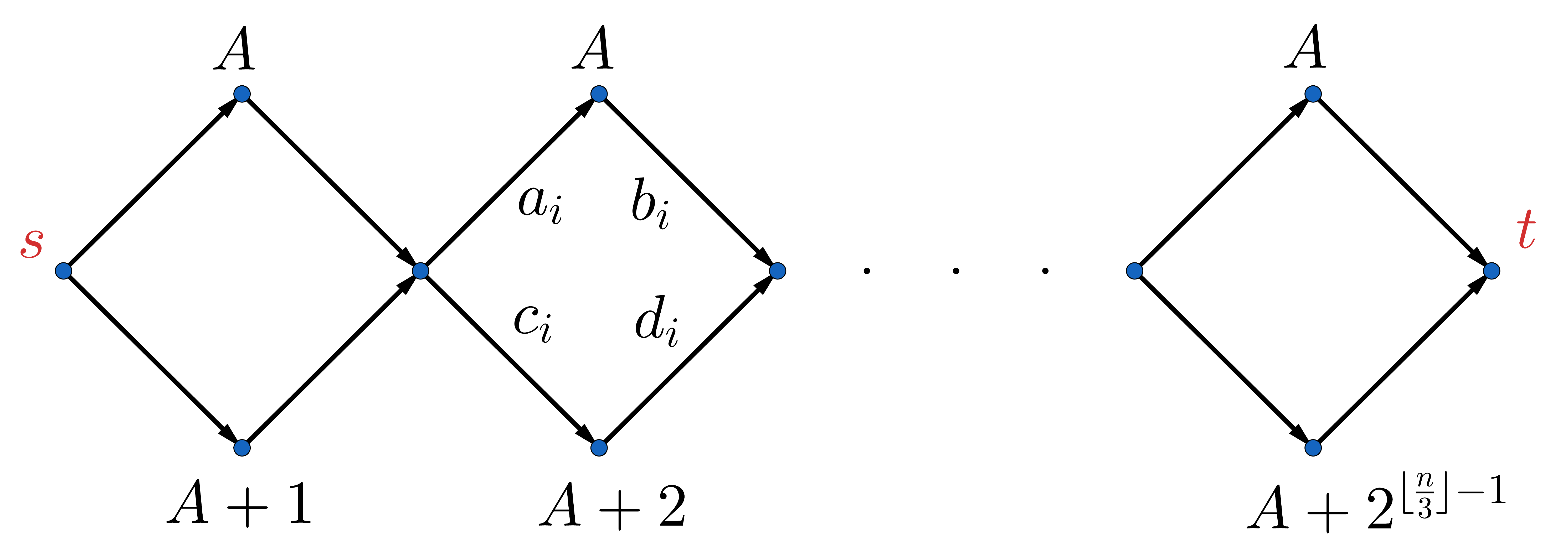}}
\caption{Exponential number of different cost of irrationality}
\label{Exp-Number}
\end{figure}

What remains is to show that we can select the values of $a_i, b_i, c_i, d_i$ and $A$. For $\beta\leq 1$, these numbers should satisfy the following conditions.

\[\begin{cases}
a_i + \beta b_i = c_i + \beta d_i \\
c_i + d_i - a_i - b_i = 2^i \\
a_i + b_i = A
\end{cases} 
%\begin{cases}
%A - b_i(1-\beta)  = c_i + \beta d_i \\
%c_i + d_i = A + 2^i \\
%a_i = A - b_i
%\end{cases}
\]
The above system has infinitely many solutions. 
%Indeed, 
%substituting into the first  equation   $c_i + d_i$ from equation 2 and canceling the value $ A $ on both sides, we obtain:
%
%$$\begin{cases}
%-b_i(1-\beta)  = 2^i - d_i(1-\beta) \\
%c_i = A - d_i + 2^i \\
%a_i = A - b_i
%\end{cases}
%\begin{cases}
%d_i = b_i + \dfrac{2^i}{1-\beta} \\
%c_i = A - d_i + 2^i \\
%a_i = A - b_i
%\end{cases}
%$$
For example, a solution to such a system is any set of the form:
$$\begin{cases}
a_i = A-1 \\
b_i = 1 \\
d_i = 1 + \dfrac{2^i}{1-\beta} \\
c_i = A - 2^i \cdot \dfrac{\beta}{1-\beta} - 1
\end{cases}$$
We have to be a bit careful to select the constant A  in such a way that the weights of all edges are non-negative. Thus we put
$$ A > b + 2^i \cdot \dfrac{\beta}{1-\beta} \implies A > b + 2^{\frac{n}{3}} \cdot \dfrac{\beta}{1-\beta}.$$

\end{proof}

\begin{remark}
By Theorem \ref{thm:expnumpaths},   the difference between the costs of the minimum and maximum feasible paths in the graph can be exponential from the number of vertices.
\end{remark}

%%%%%%%%%%%%%%%%%%%%% estimating_cost_irrat
\section{Estimating the cost of irrationality}\label{sec:estcost}
In this section, we will evaluate the complexity of estimating a random value $\costirrationality$. We give a parsimonious reduction of the following problem to  \EstIrrCostshort.

%Recall that $C_\beta$ is a random variable equal to the agent's costs in the given model $M$.

%We remind  \classSharpP-hard of the partition problem, which we will reduce to EIC problem.

\defsimpleproblem{\CountPar}{Set of positive integers $S = \{s_1,\ldots,s_n\}$.} {Count the number of partitions of $S$ into sets $S_1$ and $S_2$ such that the sums of numbers in both sets are equal.}
 
\CountPar is known to  be \classSharpP-hard \cite{DyerFKKPV93}.

%\paragraph{\bf The partition problem}
%Input: Set of positive integers $S = \{s_1,\ldots,s_n\}$.
%Task: find the number of partitions of $S$ into sets $S_1$ and $S_2$ such that the sums of numbers in both sets are equal.

\begin{theorem}\label{theorem_ECIhard}
    The  \EstIrrCostshort  problem is  \classSharpP-hard.
\end{theorem}

\begin{proof}
Let's reduce the \CountPar to our problem. For an instance $S = \{s_1,\ldots,s_n\}$ 
we construct a time-inconsistent planning model $M = (G, w, s, t, p,  \beta)$. Every  $s$-$t$ paths in $G$ will be  feasible and there will be a bijection between feasible paths of certain cost in $G$ and partitions of the set $S$ into two parts. Thus the number of feasible paths will be the solution to 
 \CountPar.
 
  Our construction works for any present bias $\beta<1$.   Consider a graph consisting of ``diamond'' gadgets.  The diamond consists of $2$ vertices connected by $2$ paths of length $2$, see Fig.~\ref{fig_gadget}. The graph consists of $n$ diamonds $D_1, \dots, D_n$ concatenated together. 
  The weights of the edges are defined as follows. Let $W$ be an integer that is greater than all $s_i$. For every $i\in \{1,\dots, n\}$, 
  the edges of the first path of the diamond $D_i$ obtain weights  $s_i$ and  $\frac{W-s_i}{\beta}$. The edges of the second  path of the diamond $D_i$ obtain weights  $-s_i$ and  $\frac{W+s_i}{\beta}$. This completes the construction of $G$.
  
  We also add that we can get rid of the negative weights of the edges by adding the same additive to all the edges, it is easy to understand that the agent's solution will not change from this additive.
  
  Let us note that for the agent standing in the first vertex $v$ of a diamond $D_i$ there are exactly two perceived paths, the first of which starts with the upper ($s_i$)-path of the diamond $D_i$, the second with the bottom ($-s_i$)-path, and both of them continue with the upper (shortest) paths of all the remaining diamonds.
  %the   {perceived cost to the target} $\zeta_M(v)$ is realized by \emph{any} $v$-$t$ path. 
  %And thus each perceived $v$-$t$ path contains either an edge of weight $s_i$ or of weight  $-s_i$. 
  In the Markov decision process, we assume that the agent select one of these edges with probability $p=1/2$. Since each of the $s$-$t$ paths in model $M$ is feasible, each of these paths will be used with probability $(1/2)^n$.
  
    \begin{figure}[ht]
    \center{\includegraphics[scale=0.18]{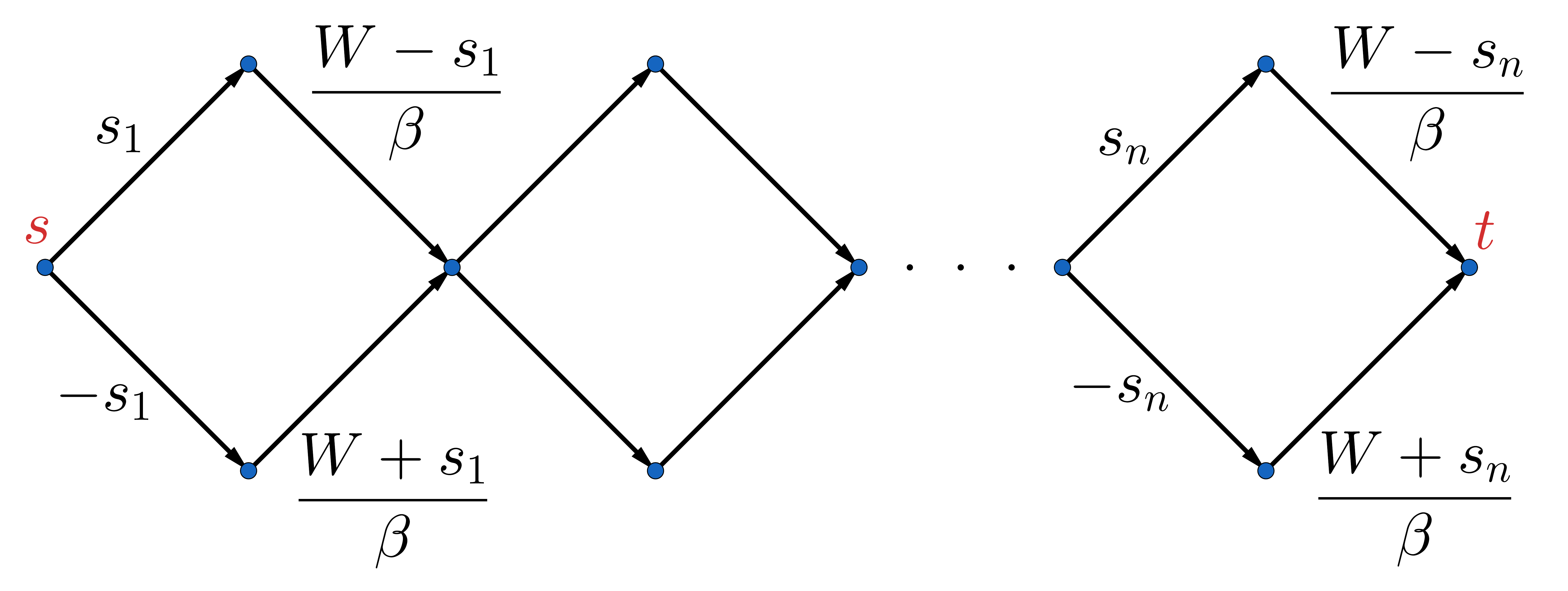}}
    \caption{Gadget used in the proof of the theorem.}\label{fig_gadget}
    \end{figure}
%    For the first edge $s_i$, the second edge will have the cost $\frac{W-s_i}{\beta}$, and for the first edge $-s_i$, the cost of the second edge will be $\frac{W+s_i}{\beta}$. Then the agent can pass through any of the first two edges in each gadget, i.e. the agent can pass through all possible paths in the graph.
%    

    Now let us show the bijection between feasible paths of cost $\frac{n\cdot W}{\beta}$ and 
   equal 
    partitions of $S$.  
    
    In one direction, let  $A$ and $B$ be a partition of $S$ such that 
    $\sum_{s\in A} s=\sum_{s\in B} s$. 
We take the path  corresponding to this  partition. When passing through diamond $D_i$, the path goes through  edge of weight $s_i$ is $s_i\in A$ and $-s_i$ otherwise.
We define \[
    \begin{cases}
   \delta= 0,\text{ if } s_i \in A \\
    \delta = 1,\text{ if } s_i \in B.
    \end{cases}
    \]
  Then the total cost of such a path is equal to
    $$
	\sum_{i=1}^n (-1)^\delta \cdot s_i + \sum_{i=1}^n \frac{W + (-1)^{\delta+1} \cdot s_i}{\beta} = \frac{n \cdot W}{\beta}.
    $$
 
% We have removed a lot of beautiful results, but we didn't have a choice in this difficult situation.
    
  For the opposite direction.   Let $P$ be a path of cost $\frac{n\cdot W}{\beta}$. The way $P$ traverses through each of the diamonds,  specify a partition of $S$ into two sets $A$ and $B$. We claim that  $\sum_{s\in A} s=\sum_{s\in B} s$. Targeting towards a contradiction, assume that
  $Q=\sum_{s\in A} s-\sum_{s\in B} s>0$. (The arguments for $Q=\sum_{s\in A} s-\sum_{s\in B} s<0$ are similar.)
  
%  Let 
%  \[
% \begin{cases}
%  \delta= 0,\text{ if } s_i \in A \\
%    \delta = 1,\text{ if } s_i \in B.
%    \end{cases}
%  \]
  Then the cost of $P$  is equal to
    \begin{eqnarray*}
    \sum_{i=1}^n (-1)^\delta \cdot s_i &+& \sum_{i=1}^n \frac{W + (-1)^{\delta+1} \cdot s_i}{\beta}  \\ &=&  Q + \frac{n \cdot W - Q}{\beta} 
     < \frac{n \cdot W }{\beta}. 
   \end{eqnarray*}
         But this is a contradiction to our assumption that the cost of $P$ is $\frac{n \cdot W }{\beta}$.
         
        % since $Q - \frac{Q}{\beta}$ is not zero, then the cost of found path not equals to $\frac{n\cdot W}{\beta}$.
    
  We have constructed a parsimonious reduction of the partitioning problem to the problem of counting the number of feasible paths of   cost $\frac{n \cdot W }{\beta}$. Thus, by counting the number of different feasible paths of  cost  $\frac{n \cdot W}{\beta}$, we can count the number of different solutions for the partition problem.

   Let $T= \frac{n \cdot W}{\beta}$. 
 We already established that counting the number of $s$-$t$ paths of cost $T$ in $G$ is    \classSharpP-hard. Now  
   we show  that computing $\Pr(C_{\beta} \leq T)$ is 
   \classSharpP-hard. Note that in our graph all paths are feasible to the agent. Thus each of the paths will be traversed by the agent with the same probability $(\frac{1}{2})^n$.
   %, in this case we get a uniform distribution on all paths. 
   Let 
   $P_{\leq T}$ be the number of paths of length at most $T$ and  $P_{=T}$ be the number of paths of length exactly  $T$. 
   Then 
    \begin{eqnarray*}
    \Pr(C_{\beta} \leq T) &=& \frac{P_{\leq T}}{2^n} = \frac{P_{=T} + P_{\leq T-1}}{2^n}\\ &=& \frac{P_T}{2^n} + \Pr(C_{\beta} \leq T-1) .
    \end{eqnarray*}
Therefore,  the existence of a  polynomial time algorithm computing $\Pr(C_{\beta} \leq T)$ would allow us to count in polynomial time the number of paths of cost $T$. 
   
 Finally, let us remind that   $\costirrationality=  \frac{C_\beta}{d(s,t)}$. Since the minimum cost 
 $d(s,t)$ is computable in polynomial time by making use of the Bellman-Ford algorithm, a polynomial time algorithm computing   $\Pr(\costirrationality\leq  T/d(s,t))$ would allow to compute in polynomial time $Pr(C_{\beta} \leq T)$, which is \classSharpP-hard.
% 
% Finally, let $R=T/d(s,t)$. Let us note that 
%    where $P_{\leq W}$ is the number of paths of length no more than $W$, $P_W$ is the number of paths of length $W$. Thus, if we calculate $\Pr[C_{\beta} \leq W]$ and $\Pr[C_{\beta} \leq W-1]$, we immediately get the number of paths of length $W$, so the task of calculating the probability of  \classSharpP-hard. 
%
%    
%    
%    Now, using the problem of calculating the probability of $\Pr[\costirrationality \leq W]$, which is similar to calculating the probability of $\Pr[C_{\beta} \leq W]$, the difference is only in multiplying by the shortest path that we can find for $O(n^3)$. We will find the number of paths of length $W$. Note that in our graph all paths are feasible to the agent and have the same probability, in this case we get a uniform distribution on all paths. Then $\Pr[C_{\beta} \leq W] = \frac{P_{\leq W}}{2^n} = \frac{P_W + P_{\leq W-1}}{2^n} = \frac{P_W}{2^n} + \Pr[C_{\beta} \leq W-1]$, where $P_{\leq W}$ is the number of paths of length no more than $W$, $P_W$ is the number of paths of length $W$. Thus, if we calculate $\Pr[C_{\beta} \leq W]$ and $\Pr[C_{\beta} \leq W-1]$, we immediately get the number of paths of length $W$, so the task of calculating the probability of  \classSharpP-hard.
\end{proof}

Although in general  \EstIrrCostshort appears to be a difficult problem, in some interesting cases described below, it can be solved in polynomial time. We   define the following two ``extremal'' cases of the problem. In the first one we estimate the probability that the agent will follow one of the feasible paths of minimum cost. The second is to compute the maximum cost of a feasible path. 

 \defsimpleproblem{\MinIrrCost (\MCI) }{A time-inconsistent planning model $M = (G, w, s, t, p,  \beta)$.} {Compute  the minimum value $W$ such that $\Pr(\costirrationality \leq W) > 0$ and compute $\Pr(\costirrationality \leq W)$. }
 
  \defsimpleproblem{\MaxIrrCost}{A time-inconsistent planning model $M = (G, w, s, t, p,  \beta)$.} {Compute  the minimum value $W$ such that $\Pr(\costirrationality \leq W) =1$.}
 
%
%
%Input: $M = (G, w, s, t, p,  \beta)$.
%Task: find out the cost $W$ of the minimum feasible path ($\min \ W :\ \Pr[C_\beta \leq W] > 0$) and compute the probability that cost of agent's path is equal to $W$.
%\end{problem}

\begin{theorem}\label{lemma:MinCI}
\MinIrrCost and  \MaxIrrCost admits an algorithm with running time $\Oh (n^3)$.
\end{theorem}

\begin{proof}
We show how to compute in polynomial time     the minimum value $W$ such that $\Pr(C_\beta \leq W) > 0$ and how to compute $\Pr(C_\beta \leq W) $. Since  $\costirrationality=  \frac{C_\beta}{d(s,t)}$ and $d(s,t)$ is computable in polynomial time, this also would yield a polynomial time algorithm solving \MCI.  The algorithm for \MaxIrrCost is similar, and we do not provide it here. 

We will traverse the vertices in the order of their topological sorting and for each of them calculate two values: $D_v$, the length of the shortest path from $s$ to $v$ that the agent chooses, and $P_v$,  the probability that the agent arrived at the vertex $v$ along the path of cost $D_v$. See Algorithm~\ref{alg:MCR}.

It is possible to make a topological sorting of the vertices of the graph $G$ and obtain an array $A$ in time $\Oh(n^2)$. Note that, having first counted the shortest paths in the graph between any pair of vertices in time $\Oh(n^3)$, for each vertex $v$ we can find the set $U$ from line 5 of the Algorithm~\ref{alg:MCR} in time $\Oh(n^2)$, going through all ancestors of $v$, and for each of them by modeling the agent's estimate in linear time.

In line 6 of our algorithm, for each vertex, the minimum path cost $D_v$ will be calculated, since we take the minimum over all vertices from which the agent can get to $v$.
And in line 7 all neighbors of $v$ are considered, passing through which the agent can achieve the minimum cost $D_v$, calculated in the previous line. Since the events of arrival at the vertex from different neighbors are inconsistent, the total probability of getting to the vertex $v$, having spent $D_v$, is calculated as the sum of the probabilities for all such neighbors.

Thus, we get the total running time of the algorithm equal to $\Oh(n^3)$.

\begin{algorithm}[tb]
\caption{Dynamic programming for \MCI}
\label{alg:MCR}
\textbf{Input}: $M = (G, w, s, t, p,  \beta)$\\
%\textbf{Parameter}: Optional list of parameters\\
\textbf{Output}: $W,\ \Pr(C_\beta \leq W)$
\begin{algorithmic}[1] %[1] enables line numbers
\STATE Let $A$ -- an array of topologically sorted vertices, $s = A[0],\ t = A[n]$
\STATE $D_s = 0$
\STATE $P_s = 1$
\FOR{$v \in A$}
\STATE $U$ - the set of neighbors of the vertex $v$, such that from $u$ the agent can go to $v$
\STATE $D_v := \min_{u \in U} (w(u,v) + D_u)$
\STATE $P_v := \sum_{u \in U'} P_u \cdot \Pr(u \to v)$, where 
$U' = \argmin_u (w(u,v) + D_u)$
\ENDFOR
\STATE \textbf{return} $D_t,\ P_t$
\end{algorithmic}
\end{algorithm}

%\todo[inline]{Explain the algorithm}

\end{proof}

%
% 
% The proof of the following lemma  is similar to the proof of Lemma~\ref{lemma:MinCI}.
%\begin{lemma}\label{lemma:MaxCI}
%\MaxIrrCost admits an algorithm with running time $\Oh (n^3)$.
%\end{lemma}
% \paragraph{\bf Computing cost of irrationality (\EstIrrCostshort ) }
% Input: $M = (G, w, s, t, p,  \beta)$ and path cost $W$.
% Task: compute the probability that cost of agent's path is not higher than $W$, $\Pr[C_\beta \leq W]$.

Finally we prove that if the weights of edges are polynomial in $n$, then  \EstIrrCostshort is solvable in polynomial time. 

\begin{theorem}
 \EstIrrCostshort  admits an algorithm with running time $\Oh( \lfloor W \cdot d(s,t) \rfloor \cdot n^2 + n^3)$.
\end{theorem}

\begin{proof}
We will traverse the vertices in the order of their topological sorting. For each vertex $v$, we will calculate the array $P_v$, numbered $0, \ldots, \lfloor W \cdot d(s,t) \rfloor$, where cell $P_v[k]$ will store the probability that the agent arrived at the vertex $v$ along the path of cost k. See Algorithm~\ref{alg:ECI}.

It is possible to make a topological sorting of the vertices of the graph $G$ and obtain an array $A$ in time $\Oh (n^2)$. Note that having first counted the shortest paths in the graph between any pair of vertices in time $\Oh(n^3)$, for each vertex $v$ we can find the set $U$ from line 5 of the Algorithm~\ref{alg:ECI} in time $\Oh(n^2)$, going through all ancestors of $v$, and for each of them by modeling the agent's estimate in linear time.

Note that in line 7 of our algorithm, for each vertex $v$ and each possible cost of the path $k$, the probability that the agent will arrive at the vertex $v$ along the path of cost $k$ is correctly calculated, since the events of arrival at the vertex from different neighbors are inconsistent and the total probability of getting to the vertex $v$ is calculated as the sum of the probabilities for all available neighbors.

So the total running time of the algorithm is $\Oh (n^3) + \Oh ( \lfloor W \cdot d(s,t) \rfloor \cdot n^2)$.

\begin{algorithm}[tb]
\caption{Dynamic programming for  \EstIrrCostshort }
\label{alg:ECI}
\textbf{Input}: $M = (G, w, s, t, p,  \beta)$, $W \geq 0$\\
%\textbf{Parameter}: Optional list of parameters\\
\textbf{Output}: $\Pr( C_\beta \leq \lfloor W \cdot d(s,t)\rfloor \, )$
\begin{algorithmic}[1] %[1] enables line numbers
\STATE Let $A$ -- an array of topologically sorted vertices, $s = A[0],\ t = A[n]$
\STATE $P_s = [1, 0, 0, \dots , 0]$
\STATE $P_v = [0, 0, 0, \dots , 0]\ \forall v \neq s$
\FOR{$v \in A$}
\STATE $U$ - the set of neighbors of the vertex $v$, such that from $u$ the agent can go to $v$
\FOR{$k := 0$ to $\lfloor W \cdot d(s,t)\rfloor $}
\STATE $P_v[k] := \sum_{U'} P_u[k - w(u,v)] \cdot \Pr(u \to v)$, where $U' = \{u \in U \ |\ w(uv) \leq k\}$
\ENDFOR
\ENDFOR
\STATE \textbf{return} $\sum_{k=0}^{\lfloor W \cdot d(s,t)\rfloor } P_t[k]$
\end{algorithmic}
\end{algorithm}

%\todo[inline]{Explain the algorithm}
\end{proof}

%%%%%%%%%%%%%%%%%%%%% expected_cost
\section{Computing expected cost}\label{sec:expcost}
We proved that the \EstIrrCostshort problem is $\# P$-hard. In this section we show that computing the expectation $\Em(\costirrationality)$ and the variance $\Var(\costirrationality)$ of  random variable $\costirrationality$ can be done in polynomial time.

% Consider a probabilistic model in which an agent, if there is an equivalent choice between several paths, chooses them equally likely or according to the given distribution. Then we could consider the problem of finding the mathematical expectation and variance of the cost of the path along which the agent will pass.
%
%\begin{problem}\label{Ecost}
%Input: $M = (G, w, s, t, p,  \beta)$. 
%
%Output: the mathematical expectation of the cost of the path chosen by the agent.
%\end{problem}
%
%\begin{problem}\label{Dcost}
%Input: $M = (G, w, s, t, p,  \beta)$. 
%
%Output: variance of the cost of the path chosen by the agent.
%\end{problem}
%
%\begin{lemma}
%The problem \ref{Ecost} admits an algorithm of complexity $O (n^3)$.
%\end{lemma}
\begin{theorem}\label{thm:expect}
For the input $M = (G, w, s, t, p,  \beta)$ of  \EstIrrCostshort, the values 
 $\Em(\costirrationality)$ and  $\Var(\costirrationality)$  are computable in time
  $\Oh (n^3)$.
\end{theorem}

\begin{proof}
We will calculate the mean and variance for the random variable $C_\beta$, and then we can easily obtain the desired values by following: $\Em(X_\beta) = \frac{\Em(C_\beta)}{d(s,t)}$ and $\Var(X_\beta) = \frac{\Var(C_\beta)}{d^2(s,t)}$.

We will calculate sequentially for each vertex the values of three quantities. To determine them, let's fix a vertex $v$ and define a random variable $X$ as the cost of the path traveled by the agent to a vertex $v$. Then, for this vertex, the values will be determined as follows:
\begin{itemize}
    \item the probability that the agent will arrive at the $v$: $\Pr(\text{come to}\ v)$;
    \item $E_v := \Em (X)$;
    \item $D_v := \Em (X^2)$.
\end{itemize}
Note that with such a definition at the vertex $t$ the random variable $X$ will coincide with $C_\beta$, which we are interested in.
See the Algorithm~\ref{alg:expectation}.

\begin{algorithm}[tb]
\caption{Dynamic programming for  $\Em$ and $\Var$ }
\label{alg:expectation}
\textbf{Input}: $M = (G, w, s, t, p,  \beta)$\\
%\textbf{Parameter}: Optional list of parameters\\
\textbf{Output}: $\Em(\costirrationality)$, $\Var(\costirrationality)$
\begin{algorithmic}[1] %[1] enables line numbers
\STATE Let $A$ -- an array of topologically sorted vertices, $s = A[0],\ t = A[n]$
\STATE $\Pr[\text{come to}\ s] = 1$
\STATE $E_s = 0$
\STATE $D_s = 0$
\FOR{$v \in A$}
\STATE $U$ - the set of neighbors of the vertex $v$, such that from $u$ the agent can go to $v$
\STATE $\Pr(\text{come to}\ v) = \sum_{u \in U} \Pr(\text{come to}\ u) \cdot \Pr(u \rightarrow v)$
\STATE $E_v =  \sum_{u \in U} \Pr(u \rightarrow v) \cdot E_u\ +\ w(uv) \cdot \Pr(u \rightarrow v) \cdot \Pr(\text{come to}\ u)$
\STATE $D_v = \sum_{u \in U} \Pr(u \rightarrow v) \cdot D_u\ +\ w^2(uv) \cdot \Pr(u \rightarrow v) \cdot \Pr(\text{come to}\ u)\ +$

\hfill$+\ 2 E_u \cdot w(uv) \cdot \Pr(u \rightarrow v)$
\ENDFOR
\STATE \textbf{return} $E_t$, $D_t - E_t^2$
\end{algorithmic}
\end{algorithm}

It is possible to make a topological sorting of the vertices of the graph $G$ and obtain an array $A$ in time $\Oh (n^2)$.
Note that, having first counted the shortest paths in the graph between any pair of vertices in time $\Oh(n^3)$, for each vertex $v$ we can find the set $U$ from line 6 of the Algorithm~\ref{alg:expectation} in time $\Oh(n^2)$, going through all ancestors of $v$, and for each of them by modeling the agent's estimate in linear time.

It is easy to see that the probability to come to the vertex $v$ is correctly calculated in line 7 of our algorithm, since it is presented as a sum of inconsistent events --- arrivals from different ancestors of the vertex.
The validity of the expression written in line 8 immediately follows from the representation of the mathematical expectation: $\Em(X) = \sum_{\text{path}} \Pr(\text{path}) \cdot \text{Cost}({\text{path}}) $. 
To obtain the formula from line 9 we use the definition of the mathematical expectation of a random variable $X^2:\ \Em(X^2) = \sum_{\text{path}} \Pr(\text{path}) \cdot \text{Cost}^2(\text{path})$, and then we split each $s$-$v$ path into two pieces: $s$-$u$ + $u$-$v$.
And finally, we get the variance of the random variable using the following well-known formula: $\Var (C_\beta) = \Em (C_\beta^2) - \Em^2 (C_\beta)$.

Thus, we get the total running time of the algorithm equal to $\Oh(n^3)$.
\end{proof}

The algorithms for the mean and the variance could be useful to \emph{motivate} the agent. 
Let us consider the situation when for   { time-inconsistent planning model} $M = (G, w, s, t, p, \beta)$, we can  choose a reward to motivate an agent to achieve a goal (target vertex). 
At every step, the  agent decides whether he wants to proceed further based on the following estimations. 
%finds the paths with the minimum cost, taking into account the present bias, selects one path from them in accordance with some probability distribution. 
The agent   compares the perceived cost of the remaining tasks, taking into account the present bias, and the reward: if the reward is greater than the estimate of the remaining path, then the agent moves further, otherwise he \emph{abandon}  his attempts to reach the goal. Then the natural algorithmic question in time-inconsistent planning \cite{KleinbergO14}, is how to identify the minimum reward that will allow the agent to reach his goal? 

With the mean and variance, we can estimate the minimum award that can help to avoid abandonment. We need the Chebyshev's inequality:
$$\Pr(|C_\beta - \Em(C_\beta)| \geq a) \leq \frac{\Var(C_\beta)}{a^2} .$$
For  $a =2 \sqrt{\Var(C_\beta)}$, we have 
$$\Pr(|C_\beta - \Em(C_\beta)| \leq 2 \sqrt{\Var(C_\beta)}) \geq \frac{3}{4}.
$$
Then, as a reward, we  take $\Em(C_\beta) + 2 \sqrt{\Var(C_\beta)}$.
% for the cost of the path, it is suitable for the reward, since the agent's estimates at each vertex are less than the cost of the path along which the agent will pass. 
 With such reward the probability that the agent will reach his goal is  at least $3/4$.

One can also consider a model in which the costs that the agent already has spent  are deducted from the reward. (Imagine the situation when the agent has some resources and will not reach the goal when the resources are exhausted). In this case, the reward must be at least  the cost of the perceived path. In this situation, the reward provided by the Chebyshev's inequality is optimal for a probability $\frac{3}{4}$ of reaching the goal.

%%%%%%%%%%%%%%%%%%%%% parameterized
\section{Parameterized complexity of \EstIrrCostshort}\label{sectionParCom}

 In this section we investigate parameterized complexity of \EstIrrCost.
 A \emph{parameterized problem} is a language $Q\subseteq \Sigma^*\times\mathbb{N}$ where $\Sigma^*$ is the set of strings over a finite alphabet $\Sigma$. Respectively, an input  of $Q$ is a pair $(I,k)$ where $I\subseteq \Sigma^*$ and $k\in\mathbb{N}$; $k$ is the \emph{parameter} of  the problem. 
A parameterized problem $Q$ is \emph{fixed-parameter tractable} (\classFPT) if it can be decided whether $(I,k)\in Q$ in time $f(k)\cdot|I|^{\Oh(1)}$ for some function $f$ that depends of the parameter $k$ only. Respectively, the parameterized complexity class \classFPT is composed by  fixed-parameter tractable problems.
The $\operatorClassW$-hierarchy is a collection of computational complexity classes: we omit the technical
definitions here. The following relation is known amongst the classes in the $\operatorClassW$-hierarchy:
$\classFPT=\classW{0}\subseteq \classW{1}\subseteq \classW{2}\subseteq \ldots \subseteq \classW{P}$. It is widely believed that $\classFPT\neq \classW{1}$, and hence if a
problem is hard for the class $\classW{i}$ (for any $i\geq 1$) then it is considered to be fixed-parameter intractable. For our purposes,
to prove that a problem is   $\classW{1}$-hard it is sufficient to show that an \classFPT algorithm for this problem yields an  \classFPT algorithm for some  $\classW{1}$-complete problem. 
We refer to    \cite{CyganFKLMPPS15} for the detailed introduction to parameterized complexity.  
 
In graph algorithms, one of the most popular parameter is the treewidth of (an undirected) graph. 
Many \classNP-hard problems are  \classFPT parameterized by the treewidth of the input graph. 
We  refer  \cite{CyganFKLMPPS15}  for the definition of treewidth. For directed graph $G$, let  $\tw(G)$ be  the treewidth of its underlying undiricted graph.  In Theorem~\ref{theorem_ECIhard}, we have proved that  \EstIrrCostshort    is $\# P$-hard. The underlying undirected graph used in the reduction in Theorem~\ref{theorem_ECIhard}, has treewidth at most $2$. (See Fig.~\ref{fig_gadget}). Thus we immediately obtain the following corollary. 

\begin{corollary}\label{cor_ECIhardtw}
    The  \EstIrrCostshort  problem remains  $\# P$-hard even when the graph $G$ in  the time-inconsistent model has $\tw(G)\leq 2$.
\end{corollary}

Besides the treewidth of the graph, another popular in the literature  graph  parameters are vertex cover and feedback vertex set. Let us remain, that for an undirected graph $G$, a \emph{vertex cover} of $G$ is a set of vertices $S\subseteq V(G)$, such that every edge of $G$ has at least one endpoint in $S$. In other words, the graph $G-S$ has no edges. For directed graph $G$, we use $\vc(G)$ to denote the minimum size of a vertex cover in the underlying undirected graph of $G$. 
A \emph{feedback vertex set} of an undirected graph $G$ is the set of vertices $S$ such that every cycle in $G$ contains at least one vertex from $S$. In other words,   graph $G-S$  
has no cycles and thus is a forest.   For directed graph $G$, we use $\fvs(G)$ to denote the minimum size of a feedback vertex cover in the underlying undirected graph of $G$. Let us note that 
\[
 \tw(G)\leq \fvs(G) \leq \vc(G)  
.\]

In what follows, we prove that   \EstIrrCostshort  is $\classW{1}$-hard  
parameterized by  $\vc(G)$. Since $\fvs(G) \leq \vc(G)$, it also yields that  \EstIrrCostshort  is \classW{1}-hard  
parameterized by  $\fvs(G)$.  On the other hand, we will give an algorithm solving   \EstIrrCostshort  
in time $n^{\Oh(\fvs(G))}$. Thus the problem is \classXP parameterized by $\fvs(G)$. Since 
$ \fvs(G) \leq \vc(G)$ it also implies that the problem is  \classXP parameterized by $\vc(G)$. 
 
We start from the lower bound.  We reduce   the following \classW{1}-hard problem to \EstIrrCostshort. 
\medskip

\defparproblem{\textsc{Modified $k$-Sum}}{Sets of integers $X_1, X_2, \dots, X_k$ and  integer $T$}{$k$}{Decide whether there is $x_1\in X_1$, $x_2\in X_2$, $\dots,$ $x_k\in X_k$ such that $x_1+\ldots+x_k=T$.}

%\paragraph{\textsc Modified $k$-Sum}
%Input: sets $X_1, X_2, \dots, X_k$, integer $T$.
%Task: decide whether there is $x_1\in X_1, x_2\in X_2, \dots, x_k\in X_k$ such that $x_1+\ldots+x_k=T$.

The following lemma is a folklore. We could not find its proof in the literature and  provide a sketch.
\begin{lemma}
The \textsc{Modified k-Sum} problem is \classW{1}-hard with respect to the parameter $k$.
\end{lemma}

\begin{proof}[Sketch] 
Suppose that there is an \classFPT algorithm solving \textsc{Modified k-Sum}. We show that then \textsc{k-Sum} is also \classFPT and this will contradict the \classW{1}-hardness of \textsc{k-Sum}, see \cite{DowneyF99}. Recall, that in the instance of \textsc{k-Sum} we have a set $S$ of 
integers $a_1, a_2, \dots, a_n$, integer  $T$, and an integer parameter $k$. The task is to decide wether  the sum of $k$ integers from $S$ is equal to $T$. 
For instance  
$a_1, a_2, \dots, a_n, T$ of \textsc{k-Sum} we apply the color coding technique \cite{AlonYZ}. 

If the instance of \textsc{k-Sum} has a solution (a subset of size $k$ summing to $T$), then if we color the numbers in $k$ colors, with a probability $\frac{k!}{k^k}$, in our solution all numbers will receive different colors.  So we color the set of numbers and from the colored numbers we form the instance of the modified \textsc{k-Sum}, that is, each color corresponds to set $X_i$, and call the \classFPT algorithm on it. There are also well-known approaches how to  derandomize the color-coding. Thus we solved \textsc{k-Sum} in \classFPT time, which yields that $\classFPT= \classW{1}$. 
\end{proof}

\begin{theorem}\label{thm:w1-hard}
The \EstIrrCostshort problem is \classW{1}-hard  %with the parameter $\fvs(G)$ and
parameterized by  $\vc(G)$ and by $\fvs(G)$.
\end{theorem}

\begin{proof}
We construct a parameterized reduction of the \textsc{Modified k-Sum} problem to the \EstIrrCostshort  problem.

Let's give an instance of the \textsc{Modified k-Sum} problem: $X_1, X_2, \dots, X_k$ and $T$.
Similarly to the proof of \classSharpP hardness, we construct for each set $X_i$ a gadget in which all paths will be perceived for the agent and will be evaluated in $W$ (for an edge $c_i$, an additional edge will have the weight $\frac{W-c_i} {\beta}$), where $W$ is an integer greater than all $x \in \cup_i X_i$. The graph $G$ consists of $k$ gadgets concatenated together. As $\beta$, we take any constant from 0 to 1, for example, $\beta = \frac{1}{2}$.

\begin{figure}[ht]
    \center{\includegraphics[scale=0.18]{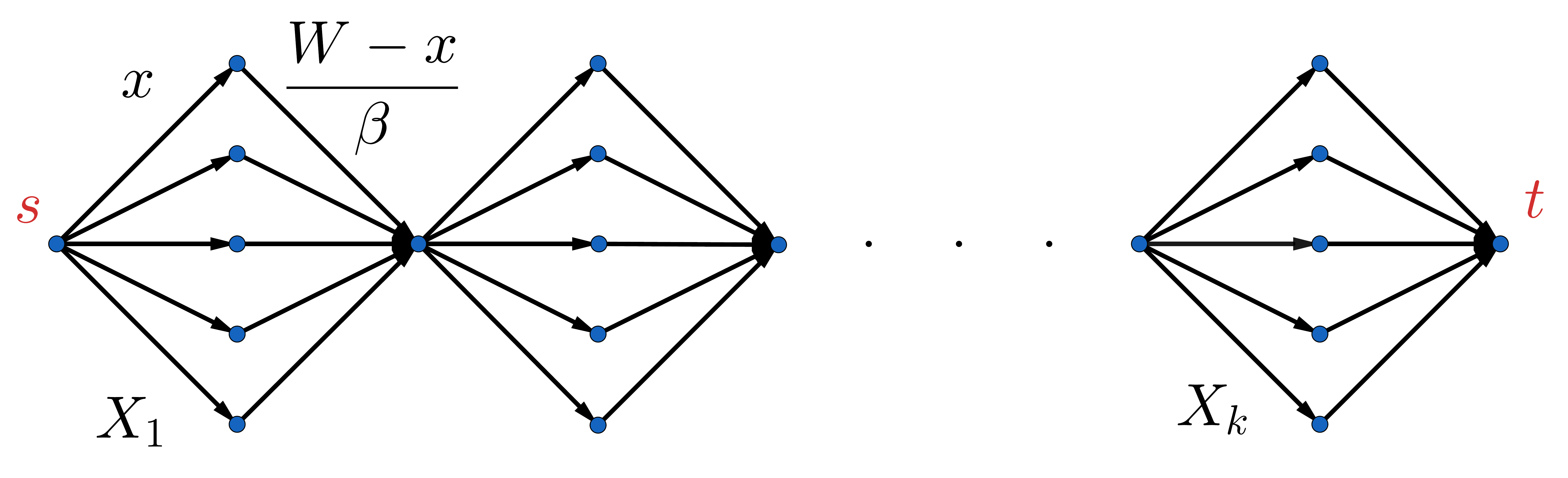}}
    \caption{Gadget used in \classW{1}-hardness proof.}
\end{figure}
Let's set the target weight of the path to the agent as follows: $Y = T \cdot (1 - \dfrac{1}{\beta}) + \dfrac{k \cdot W}{\beta}$. The task parameter is the vertex cover or feedback vertex set will be equal to $\Oh(k)$.

We now show that the answer to the \textsc{Modified k-Sum} problem is positive if and only if the answer to the \EstIrrCostshort problem is positive. Let there be a set $x_i \in X_i$, which in total gives $T$, then the agent, choosing a path in the $i$ gadget, the first edge of which has a weight of $x_i$, will get a path whose total weight is
$$\sum_i x_i + \dfrac{W - x_i}{\beta} = T + \dfrac{k \cdot W}{\beta} - \dfrac{T}{\beta}.$$

Conversely, let the agent find the path of the desired weight $Y$, denote the weight of the first edge in the $i$ gadget in this path for $x_i$. Let $\sum x_i = T'$. Then $$T \cdot (1 - \dfrac{1}{\beta}) + \dfrac{k \cdot W}{\beta} = Y  = \sum_i x_i + \dfrac{W - x_i}{\beta} = $$
$$ = T' + \dfrac{k \cdot W}{\beta} - \dfrac{T'}{\beta} = T' \cdot (1 - \dfrac{1}{\beta}) + \dfrac{k \cdot W}{\beta}.$$
We get that $T' = T$.

 We already established that existence of $s$-$t$ paths of cost $Y$ in $G$ is \classW{1}-hard. Now
   we show that computing $\Pr(C_{\beta} \leq Y)$ is 
   \classW{1}-hard. 
   %Note that in our graph all paths are feasible to the agent. Thus each of the paths will be traversed by the agent with the same probability $\frac{1}{|X_1| \cdot \ldots \cdot |X_n|}$.
   %, in this case we get a uniform distribution on all paths. 
   % Let $P_{\leq Y}$ be the number of paths of length at most $Y$ and  $P_{=Y}$ be the number of paths of length exactly  $Y$. 
   Note that
   $$
   \Pr(C_{\beta} = Y) = \Pr(C_{\beta} \leq Y) - \Pr(C_{\beta} \leq Y-1).
   $$
Thus, if $\Pr(C_{\beta} = Y) > 0$, then there is a feasible path of cost $Y$.
Therefore,  the existence of a  \classFPT algorithm computing $\Pr(C_{\beta} \leq Y)$ would allow us to check for existence the paths of cost $Y$ in \classFPT time.
% $$\Pr(C_{\beta} \leq Y) = \frac{P_{\leq Y}}{|X_1| \cdot \ldots \cdot |X_n|}=$$
% $$=\frac{P_Y}{|X_1| \cdot \ldots \cdot |X_n|} + \Pr(C_{\beta} \leq Y-1).$$

\end{proof}

%From the resulting complexity, we will make some interesting comments. If we consider the graph from the reduction and remove the orientation in it, then treewidth and pathwidth are equal to 2. This means that when the problem is parameterized by these values, it is para-NP. Feedback vertex set ($\fvs$) and vertex cover ($\vc$) of our graph do not exceed k, this means that the problem $W[1]$-hard parameterized by $\fvs(G)$ and $\vc(G)$.
%
%There are a couple of interesting questions left. Will the prolem be XP when parameterizing $\fvs$? Will the task be FPT relative to some parameter, for example, feedback edge set ($\fes$)?

Theorem~\ref{thm:w1-hard} rules out the existence of an algorithm solving \EstIrrCostshort  in time $f(\vc(G))n^{\Oh(1)}$ for any function $f$ of $\vc(G)$ only. (Unless \classFPT=\classW{1}.) 
In what follows, we prove that when  $\fvs(G)$ is a constant, then the problem is solvable in polynomial time, that is, is in \classXP parameterized by 
$\fvs(G)$  (and hence by $\vc(G)$).

We start from the following combinatorial lemma. 
\begin{lemma}
\label{lem: fvs}
Let $\fvs(G)=k$.  Then the number of different $s$-$t$ paths in $G$ is at most $k^k n^{\Oh(k)}$.
\end{lemma}

\begin{proof}
Let's assume that for a graph $G$ $\fvs$ is equal to k, we remove the set of vertices $S$, $|S| = k$ that implements $\fvs$, then the undirected skeleton of our graph becomes a forest. Then there is at most one path from $s$ to $t$.

Note that when adding a set of $S$ remote vertices, oriented paths from $s$ to $t$ may appear in the graph. Let's estimate their number from above: in one path each vertex can be visited no more than once, since our directed graph was acyclic.
Therefore, each path can be characterized by a sequence of pieces of the path that go along the vertices of $S$, in the order in which they meet in the path. Their number is not more than $k^k$. And also for each such piece, we need to fix the entry vertex into it from $V(G) \setminus S$, and the exit vertex in $V (G) \setminus S$. We get that for one such piece, the number of pairs (input, output) is not more than $n^2$, and the pieces of the path are not more than $k$, so we get an estimate from above $n^{2k}$. Moreover, this choice uniquely sets the path for us, since $V (G) \setminus S$ is a forest, and there is no more than one path between any two vertices.

Hence the total number of paths from $ s $ to $ t $ doesn't exceed $ k^k \cdot n^{2k}.$
\end{proof}

By making use of Lemma~\ref{lem: fvs}, we prove the following theorem. 
\begin{theorem}\label{thm:ECIFVS}
The \EstIrrCostshort problem admits an algorithm of running time 
$n^{\Oh(\fvs(G))} \cdot \fvs(G)^{\fvs(G)}$.
%with the feedback vertex set ($\fvs$) parameter is $XP$.
\end{theorem}

\begin{proof}
Let $k=\fvs(G)$. We note that in the  pseudo-polynomial algorithm~\ref{alg:ECI},  there are no more path residues in each vertex dynamics than the number of different paths. Consequently, by  Lemma~\ref{lem: fvs}, we obtain that in line 7 of the algorithm~\ref{alg:ECI}, at each step, an array of size at most $n^{f(k)}$ will be maintained. So each set in dynamics is less than $n^{2k} \cdot k^k$, we get the time $\Oh(n^{2k+3} \cdot k^k)$.
\end{proof}

Because of Theorem~\ref{thm:w1-hard}, the running time provided by Theorem~\ref{thm:ECIFVS} is basically the best we can hope for. However, the \EstIrrCostshort  problem is \classFPT being parameterized by the feedback edge set of the underlying undirected graph. Let us remind, that a \emph{feedback edge set} of an undirected graph is  a set of edges  whose removal makes the graph acyclic. For directed graph $G$, we use   $\fes(G)$ to denote the minimum size of a feedback edge set of its underlying undirected graph.

First we bound the number of paths in a graph by a function of $\fes(G)$. 
\begin{lemma}
\label{lem: fes}
Let $\fes(G)=k$.  
Then the number of different $s$-$t$ paths in $G$ is at most 
  $2^{k}$.
\end{lemma}

\begin{proof}
We remove the set of edges $F$, $|F| = k$ that implements $\fes$.  Then the underlying undirected  graph of $G$ becomes a forest. In the forest  there is at most one path from $s$ to $t$.

Note that when adding a set of $F$ deleted edges, oriented paths from $s$ to $t$ may appear in the graph. Let's estimate their number from above: in one path each edge can be traversed no more than once---our directed graph was acyclic.
Since each directed path from $s$ to $t$ is uniquely defined by the set of edges from $F$ that participate in it, we have  that the number of paths does not exceed $2^k$.
Because the directed graph $G$ is acyclic, the order in which the edges selected from $F$ meet along the path is uniquely determined. And since the undirected skeleton of the graph $G \setminus F$ is a forest, there is at most  one path between each pair of consecutive edges selected from $F$.
\end{proof}

By Lemma~\ref{lem: fes}, we obtain  that \EstIrrCostshort is \classFPT parameterized by $\fes(G)$. 
\begin{theorem}\label{thm:fes}
The \EstIrrCostshort problem is solvable in time $2^{\fes(G)} \cdot poly (n)$.
\end{theorem}

\begin{proof}Set $k=\fes(G)$. 
By Lemma~\ref{lem: fes}, the number of different paths in the graph is at most $2^k$. Consequently, in line 7 of the algorithm~\ref{alg:ECI}, at each step, an array of size at most $2^k$ will be obtained.Then our pseudo-polynomial algorithm~\ref{alg:ECI} will work in time $2^k \cdot poly (n)$.
\end{proof}

%%%%%%%%%%%%%%%%%%%%% conclusion
\section{Conclusion}\label{sec:conclusion}
We introduced the new model of the cost of irrationality. For future research we present two open algorithmic questions related to our model. 

The first question concerns the motivation and establishing rewards \cite{Albers2018}. Assume that by achieving the goal $t$ the agent  hopes to receive a reward. If at some moment the perceived costs becomes larger than the reward, the agent abandons the mission. For a given probability $p$, how difficult is to compute (exactly or approximately) the minimum reward that would   allow the agent not to  abandon his mission with probability at least $p$? 

The second question is related to the question of  finding a motivating subgraph \cite{KleinbergO14,FominS20}.  We gave a polynomial time algorithm computing the expected   cost of irrationality $\Em(\costirrationality )$.   Consider the following algorithmic task:  delete  at most $k$ edges (or vertices) such that in the resulting graph the expected cost of irrationality is less than $\Em(\costirrationality )$.  
 Of course,  there is  a brute-force algorithm solving the problem in time $n^{\Oh(k)}$ by calling our polynomial-time algorithm for each of the  ${n}\choose{k}$ possibilities of deleting $k$ edges (or vertices). But  whether the problem is $\classFPT$ parameterized by $k$, is an interesting open question.

\bibliography{references}

\begin{thebibliography}{10}

\bibitem{Akerlof91}
George~A. Akerlof.
\newblock Procrastination and obedience.
\newblock {\em American Economic Review: Papers and Proceedings}, 81(2):1--19,
  1991.

\bibitem{AlbersK17}
Susanne Albers and Dennis Kraft.
\newblock On the value of penalties in time-inconsistent planning.
\newblock In {\em 44th International Colloquium on Automata, Languages, and
  Programming (ICALP)}, pages 10:1--10:12, 2017.

\bibitem{AlbersK17b}
Susanne Albers and Dennis Kraft.
\newblock The price of uncertainty in present-biased planning.
\newblock In {\em 13th Int. Conference on Web and Internet Economics (WINE)},
  pages 325--339, 2017.

\bibitem{Albers2018}
Susanne Albers and Dennis Kraft.
\newblock Motivating time-inconsistent agents: {A} computational approach.
\newblock {\em Theory Comput. Syst.}, 63(3):466--487, 2019.

\bibitem{AlonYZ}
Noga Alon, Raphael Yuster, and Uri Zwick.
\newblock Color-coding.
\newblock 42(4):844--856, 1995.

\bibitem{CyganFKLMPPS15}
Marek Cygan, Fedor~V. Fomin, Lukasz Kowalik, Daniel Lokshtanov, D{\'{a}}niel
  Marx, Marcin Pilipczuk, Michal Pilipczuk, and Saket Saurabh.
\newblock {\em Parameterized Algorithms}.
\newblock Springer, 2015.

\bibitem{DowneyF99}
Rodney~G. Downey and Michael~R. Fellows.
\newblock {\em Parameterized complexity}.
\newblock Springer-Verlag, New York, 1999.

\bibitem{DyerFKKPV93}
Martin~E. Dyer, Alan~M. Frieze, Ravi Kannan, Ajai Kapoor, Ljubomir Perkovic,
  and Umesh~V. Vazirani.
\newblock A mildly exponential time algorithm for approximating the number of
  solutions to a multidimensional knapsack problem.
\newblock {\em Comb. Probab. Comput.}, 2:271--284, 1993.
\newblock URL: \url{https://doi.org/10.1017/S0963548300000675}, \href
  {http://dx.doi.org/10.1017/S0963548300000675}
  {\path{doi:10.1017/S0963548300000675}}.

\bibitem{FominFG21}
Fedor~V. Fomin, Pierre Fraigniaud, and Petr~A. Golovach.
\newblock Present-biased optimization.
\newblock In {\em Proceedings of the Thirty-Fifth {AAAI} Conference on
  Artificial Intelligence (AAAI)}, pages 5415--5422. {AAAI} Press, 2021.
\newblock URL: \url{https://www.aaai.org/Library/AAAI/aaai21contents.php}.

\bibitem{FominS20}
Fedor~V. Fomin and Torstein J.~F. Str{\o}mme.
\newblock Time-inconsistent planning: Simple motivation is hard to find.
\newblock In {\em Proceeding of the 34th {AAAI} Conference on Artificial
  Intelligence (AAAI)}, pages 9843--9850. {AAAI} Press, 2020.
\newblock URL: \url{https://aaai.org/ojs/index.php/AAAI/article/view/6537}.

\bibitem{GravinILP16arch}
Nick Gravin, Nicole Immorlica, Brendan Lucier, and Emmanouil Pountourakis.
\newblock Procrastination with variable present bias.
\newblock {\em CoRR}, abs/1606.03062, 2016.

\bibitem{KleinbergO14}
Jon~M. Kleinberg and Sigal Oren.
\newblock Time-inconsistent planning: a computational problem in behavioral
  economics.
\newblock In {\em {ACM} Conference on Economics and Computation (EC)}, pages
  547--564, 2014.

\bibitem{KleinbergO18}
Jon~M. Kleinberg and Sigal Oren.
\newblock Time-inconsistent planning: a computational problem in behavioral
  economics.
\newblock {\em Commun. {ACM}}, 61(3):99--107, 2018.

\bibitem{KleinbergOR16}
Jon~M. Kleinberg, Sigal Oren, and Manish Raghavan.
\newblock Planning problems for sophisticated agents with present bias.
\newblock In {\em {ACM} Conference on Economics and Computation (EC)}, pages
  343--360, 2016.

\bibitem{KleinbergOR17}
Jon~M. Kleinberg, Sigal Oren, and Manish Raghavan.
\newblock Planning with multiple biases.
\newblock In {\em {ACM} Conference on Economics and Computation (EC)}, pages
  567--584, 2017.

\bibitem{Laibson1994}
David~I. Laibson.
\newblock {\em Hyperbolic Discounting and Consumption}.
\newblock PhD thesis, Massachusetts Institute of Technology, Department of
  Economics, 1994.
\newblock URL: \url{http://hdl.handle.net/1721.1/11966}.

\bibitem{McClure2004}
Samuel~M. McClure, David~I. Laibson, George Loewenstein, and Jonathan~D. Cohen.
\newblock Separate neural systems value immediate and delayed monetary rewards.
\newblock {\em Science}, 306(5695):503--507, 2004.
\newblock URL: \url{https://science.sciencemag.org/content/306/5695/503}, \href
  {http://arxiv.org/abs/https://science.sciencemag.org/content/306/5695/503.full.pdf}
  {\path{arXiv:https://science.sciencemag.org/content/306/5695/503.full.pdf}},
  \href {http://dx.doi.org/10.1126/science.1100907}
  {\path{doi:10.1126/science.1100907}}.

\bibitem{MihalakSW16}
Mat{\'{u}}s Mihal{\'{a}}k, Rastislav Sr{\'{a}}mek, and Peter Widmayer.
\newblock Approximately counting approximately-shortest paths in directed
  acyclic graphs.
\newblock {\em Theory Comput. Syst.}, 58(1):45--59, 2016.
\newblock URL: \url{https://doi.org/10.1007/s00224-014-9571-7}, \href
  {http://dx.doi.org/10.1007/s00224-014-9571-7}
  {\path{doi:10.1007/s00224-014-9571-7}}.

\bibitem{o1999doing}
Ted O'Donoghue and Matthew Rabin.
\newblock Doing it now or later.
\newblock {\em American economic review}, 89(1):103--124, 1999.

\bibitem{Samuelson1937}
Paul~A. Samuelson.
\newblock A note on measurement of utility.
\newblock {\em The Review of Economic Studies}, 4(2):155--161, 02 1937.
\newblock URL: \url{https://doi.org/10.2307/2967612}, \href
  {http://arxiv.org/abs/http://oup.prod.sis.lan/restud/article-pdf/4/2/155/4345109/4-2-155.pdf}
  {\path{arXiv:http://oup.prod.sis.lan/restud/article-pdf/4/2/155/4345109/4-2-155.pdf}},
  \href {http://dx.doi.org/10.2307/2967612} {\path{doi:10.2307/2967612}}.

\bibitem{tang2017computational}
Pingzhong Tang, Yifeng Teng, Zihe Wang, Shenke Xiao, and Yichong Xu.
\newblock Computational issues in time-inconsistent planning.
\newblock In {\em Proceedings of the 31st Conference on Artificial Intelligence
  (AAAI)}. {AAAI} Press, 2017.

\bibitem{thaler2015misbehaving}
Richard~H Thaler and LJ~Ganser.
\newblock Misbehaving: The making of behavioral economics.
\newblock 2015.

\end{thebibliography}
\bibliographystyle{plainurl}

\appendix

\end{document}